\documentclass[11pt,reqno]{amsart}
\usepackage{amssymb}

\usepackage{amsthm,amsfonts,amssymb,euscript,hyperref,color}
\usepackage{comment}
\usepackage{mathtools}
\usepackage{stmaryrd}
\usepackage{mathrsfs}
\usepackage{pifont}
\usepackage{amsfonts,amsmath,amssymb,amsthm,amscd,cancel}
\usepackage{graphicx}
\usepackage[]{geometry}
\usepackage{verbatim}
\usepackage{mathrsfs}
\usepackage{fancyhdr}
\usepackage{footnpag,footmisc}

\usepackage[displaymath]{lineno}
\theoremstyle{definition}
\newtheorem{theorem}{Theorem}
\newtheorem{lemma}{Lemma}

\newtheorem{remark}{Remark}

\allowdisplaybreaks[4]

\DeclareMathAlphabet{\mathsfsl}{OT1}{cmss}{m}{sl}
\numberwithin{equation}{section}

\newcommand{\D}{\mathrm{d}}

\def\Lb{\underline{L}}

\def\ub{\underline{u}}
\def\Cb{\underline{C}}

\newcommand{\Db}{\underline{D}}

\newcommand{\hb}{\underline{h}}

\begin{document}

\title[Instability in spherical symmetry]{A robust proof of the instability of naked singularities of a scalar field in spherical symmetry}

\author[Jue Liu]{Jue Liu}
\address{Department of Mathematics, Sun Yat-sen University\\ Guangzhou, China}
\email{liuj337@mail2.sysu.edu.cn}
\author[Junbin Li]{Junbin Li}
\address{Department of Mathematics, Sun Yat-sen University\\ Guangzhou, China}
\email{lijunbin@mail.sysu.edu.cn}

\thanks{Both authors are supported by NSFC 11501582, 11521101.}

\date{}

\maketitle

\begin{abstract}

Published in 1999, Christodoulou proved that the naked singularities of a self-gravitating scalar field are not stable in spherical symmetry and therefore the cosmic censorship conjecture is true in this context. The original proof is by contradiction and sharp estimates are obtained strictly depending on spherical symmetry. In this paper, appropriate a priori estimates for the solution are obtained. These estimates are more relaxed but sufficient for giving another robust argument in proving the instability, in particular not by contradiction. In another related paper, we are able to prove instability theorems of the spherical symmetric naked singularities under certain isotropic gravitational perturbations without symmetries. The argument given in this paper plays a central role.
\end{abstract}


\setcounter{tocdepth}{1}


\allowdisplaybreaks

\section{Introduction}

In the paper \cite{Chr99}, Christodoulou proved both the \emph{weak cosmic censorship conjecture} and the \emph{strong cosmic censorship conjecture} for spherically symmetric solutions of the Einstein equations coupled with a massless scalar field. The coupled system reads
\begin{align*}
\mathbf{Ric}_{\alpha\beta}-\frac{1}{2}\mathbf{R}g_{\alpha\beta}=\mathbf{T}_{\alpha\beta}=\nabla_\alpha\phi\nabla_\beta\phi-\frac{1}{2}g_{\alpha\beta}g^{\mu\nu}\nabla_\mu\phi\nabla_\nu\phi,
\end{align*}
which we call the Einstein-scalar field equations. The proof, which is by contradiction, contains sharp estimates which may not be easily obtained beyond spherical symmetry. In this paper, we will provide a robust proof which is not by contradiction and contains only relaxed estimates. The main advantage of this proof is that it has the potential to be extended beyond spherical symmetry.

Consider the characteristic initial value problem of the Einstein-scalar field equations in spherical symmetry. The initial data is given on a null cone $C_o$ issuing from a fixed point $o$ of the symmetry group $SO(3)$, and consists of a function $\alpha_0=\frac{\partial}{\partial r}(r\phi)\big|_{C_o}$ defining on $[0,+\infty)$, where $r$ is area radius of the orbit spherical sections of $C_o$, and $\phi$ is the scalar field function. Then what was exactly proved by Christodoulou is the following theorem.
\begin{theorem}[Christodoulou]\label{Chr99}
Let $\mathcal{E}$ be the complement of the collection of functions $\alpha_0\in BV$ whose maximal future development is either complete, or possesses a complete future null infinity and a strictly spacelike singular future boundary. Then if $\alpha_0\in\mathcal{E}$, then there exists some $f\in BV$ such that $\alpha_0+\lambda f\notin\mathcal{E}$ and has non-complete maximal future development for all $\lambda\ne0$. Moreover, if $\alpha_0+\lambda f\equiv\alpha_0'+\lambda'f'$, then $\alpha_0\equiv \alpha_0'$, $f\equiv f'$ and $\lambda=\lambda'$.
\end{theorem}
We may therefore say that the exceptional set $\mathcal{E}$ is of codimension at least $1$.  The proof can roughly be divided into three steps. Consider an arbitrary initial data $\alpha_0\in BV$. The first step is, which was shown in \cite{Chr93}, that the maximal future development of such data is complete, unless there exists a singular endpoint $e$ of the central timelike geodesic $\Gamma$ from $o$. A sharp criterion of the appearance of $e$ was also found: $\frac{2m}{r}\nrightarrow0$ when approaching $e$, where $m$ is the mass function. The second step is to understand how an apparent horizon forms. We have the following theorem also by Christodoulou.
\begin{theorem}[Christodoulou, \cite{Chr91}]\label{Chr91}
Consider the spherically symmetric solution of the Einstein-scalar field equations with initial data given on a null cone $C_o$. Let $S_1$ and $S_2$ be two sphercal sections with area radii $r_1$, $r_2$ and mass contents $m_1$, $m_2$, and $S_2$ is in the exterior to $S_1$. Denote
$$\delta=\frac{r_2}{r_1}-1.$$
Then there exists positive constants $c_0$, $c_1$ such that if $\delta\le c_0$ and
\begin{align*}2(m_2-m_1)>c_1r_2\delta\log\left(\frac{1}{\delta}\right),\end{align*}
then the incoming null cone through $S_2$ intersects the apparent horizon and enters the trapped region, the region of trapped surfaces. Moreover, there exists an event horizon and the trapped region terminates at a spacelike singular boundary.\end{theorem}

Inspired by this theorem, we then consider the future of $\Cb_e\bigcup C_o$, where $\Cb_e$ is the boundary of the causal past of $e$ and intersects $C_o$ at $s=s_e$. Then the last step, what was really proved in \cite{Chr99} is that, allowing a perturbation on $\alpha_0$, there exists a sequence $p_n\in\Gamma$ where $p_n\to e$ such that the null cone $C_{p_n}$ issuing from $p_n$ satisfies the assumptions of Theorem \ref{Chr91} at two spheres $S_{1,n}$ and $S_{2,n}$ on $C_{p_n}$, between which the distance tends to zero. The corresponding spacetimes have an event horizon and therefore possess a complete future null infinity, which verifies the \emph{weak cosmic censorship conjecture}. Moreover, the distance of $S_{1,n}$ and $S_{2,n}$ tending to zero implies that the apparent horizon issues from $e$ and therefore the future boundary of the maximal development is spacelike and singular. This verifies the \emph{strong cosmic censorship conjecture}. 

In this paper, we are going to give a new argument of this last step. We would like to give a statement of this single step. First of all, we introduce a double null coordinate $(\ub,u)$ of the quotient spacetime relative to the singular endpoint $e$ of $\Gamma$ as follows. Let $\ub$ and $u$ be optical functions, i.e., their level sets, which we denote by $\Cb_{\ub}$ and $C_u$, are incoming and outgoing null cones respectively. We take $\ub=0, u=-r$ on $\Cb_e$, and take $u=u_0$ and $\ub$ increasing towards the future on $C_o$ where $-u_0=r_0$ is the area radius of the sphere $\Cb_e\bigcap C_o$. Using this notation, we will write $\Cb_e=\Cb_0$, $C_o=C_{u_0}$. In terms of the double null coordinate $(\ub,u)$ relative to $e$, what we are going to reprove can be stated as follows.
\begin{theorem}\label{main}
Let $\mathcal{E}$ be the complement of the collection of functions $\alpha_0\in BV$ whose maximal future development is either complete or has the property, that if $e$ is the singular endpoint of $\Gamma$ and $(\ub,u)$ is the double null coordinate relative to $e$, then there exists two sequences $\delta_n\to0^+$ and $u_n\to0^-$ such that
\begin{align}\label{maintheoremcondition}
2(m-m_n)>\frac{c_1r(r-r_n)}{r_n}\log\frac{r_n}{r-r_n},\ \text{with}\ \frac{r-r_n}{r_n}\le c_0
\end{align}
where $m$ and $r$ take values at $(\ub,u)=(\delta_n,u_n)$, $m_n=m(0,u_n)$, $r_n=|u_n|$ and $c_0,c_1$ are the constants given in Theorem \ref{Chr91}. Then if $\alpha_0\in\mathcal{E}$, then there exists two functions $f_1,f_2\in BV$, such that $\alpha_0+\lambda_1f_1+\lambda_2f_2\notin\mathcal{E}$ and has non-complete maximal future development for all $\lambda_1,\lambda_2$ with $\lambda_1\ne0$ or $\lambda_2\ne0$. Moreover, if $\alpha_0+\lambda_1f_1+\lambda_2f_2\equiv\alpha_0'+\lambda_1'f'_1+\lambda_2'f'_2$, then $\alpha_0\equiv \alpha_0'$, $f_1\equiv f'_1$, $f_2\equiv f'_2$ and $\lambda_1=\lambda_1'$, $\lambda_2=\lambda_2'$.
\end{theorem}
\begin{remark}
As in \cite{Chr99}, we will also reprove this theorem where the exceptional set $\mathcal{E}$ is of codimension at least $2$, which is stronger than what we state in Theorem \ref{Chr99}. Both in Christodoulou's proof and the proof in this paper, $f_2$ is shown to be absolutely continuous, and therefore the conclusions of Theorem \ref{Chr99} hold for $\alpha_0$ being absolutely continuous, which is of course more regular than being of bounded variation.
\end{remark}

In proving this theorem, instead of using double null coordinate, Christodoulou worked in a dimensionless coordinate $(s,t)$ relative to the singular endpoint $e$:
\begin{align*}
u=u_0\mathrm{e}^{-t},\ -2r=u_0\mathrm{e}^{s-t},
\end{align*}
and the first step of the proof by contradiction is to assume that given any $\varepsilon>0$, the opposite of \eqref{maintheoremcondition}, i.e.,
\begin{align}\label{contradictionassumption}
2(m(s,t)-m(0,t))\le c_1r(s,t)s\log\left(\frac{1}{s}\right)
\end{align}
holds in $\{0\le s\le c_0\}\bigcap\{0\le\ub\le\varepsilon\}$. For the Einstein-scalar field system, the mass $m$ governs the whole system and has good monotone properties. Christodoulou was able to estimate all related geometric quantities in a sharp way in terms of $m(s,t)-m(0,t)$, which is bounded from \eqref{contradictionassumption}, and use these estimates to find some $(s_\varepsilon, t_\varepsilon)\in\{0\le s\le c_0\}\bigcap\{0\le\ub\le\varepsilon\}$ such that the opposite of \eqref{contradictionassumption} holds for some particular $(s,t)=(s_\varepsilon,t_\varepsilon)$, which is a contradiction. 

However, in order to extend this result beyond spherical symmetry, much more things need to be concerned. First of all, we need to derive suitable a priori estimates in order to establish the \emph{existence} of the solution. Second, we may not benefit from the assumptions like \eqref{contradictionassumption} which is from proof by contradiction because the mass, which is essentially the $L^2$ integral of $L\phi$ (together with the outgoing shear) over $C_u$, can no longer govern the whole system without symmetries. In addition, the estimates derived by Christodoulou is so sharp that it is not easy to extend them beyond spherical symmetry. 

In this paper,  we are able to derive a priori $L^\infty$ bounds of the geometric quantities, including $\frac{\partial}{\partial\ub}\phi$, $\frac{\partial}{\partial u}\phi$, and the derivatives of $r$ and $\Omega$ defined by $-2\Omega^2=g\left(\frac{\partial}{\partial\ub},\frac{\partial}{\partial u}\right)$. These a priori estimates are proved by bootstrap argument to hold in a region deep enough to the future. It then follows easily that the condition \eqref{maintheoremcondition} eventually holds before these estimates fail. These estimates are robust and analogues of them may hold when no symmetries are imposed. The generalization of these estimates without symmetries will also be used in proving the \emph{existence}. There is a simple way to understand the difference between two arguments, that is, the a priori estimates we derive will in particular imply that some analogue of the assumption \eqref{contradictionassumption} really holds with $c_1$ replaced by a larger constant depending on the $L^\infty$ bounds of the data on $C_{o}$ and this condition with a larger constant still implies that \eqref{maintheoremcondition} holds by Christodoulou's argument. Nevertheless, we do not need to repeat Christodoulou's argument because from the bootstrap argument, all estimates are obtained simultaneously and \eqref{maintheoremcondition} is then simply a direct conclusion.

The new argument presented in this paper can possibly be extended to the case when no symmetries are assumed. In an another paper \cite{Li-Liu} by the authors, we consider the characteristic initial value problem of the Einstein-scalar field equations, with the initial data given on two null cones intersecting at a sphere. The incoming null cone is assumed to be spherically symmetric and singular at its vertex, in the sense that $\frac{2m}{r}\nrightarrow0$ when approaching it. No symmetries are imposed on the outgoing null cone.  Then we will show that the argument presented in this paper can be directly generalized and we can also prove an instability theorem like Theorem \ref{Chr99}. We suggest the readers refer to \cite{Li-Liu} for the precise statement. Finally, we should also mention that the estimates derived in this paper, and also in \cite{Li-Liu}, share many common features with those in the work of An-Luk \cite{An-Luk} where they worked with the spacetime region deep near the vertex which is regular. Readers may also refer to \cite{Li-Liu} for some discussions about this.

\section{Double null coordinates and equations}

\subsection{Double null coordinate}

The spherically symmetric spacetime can be studied through its 2-dimensional quotient spacetime manifold with boundary $\Gamma$, the fixed point set of the $SO(3)$ action, being a timelike geodesic, which we call the central line. We use a double null coordinate $(\ub,u)$, where $\ub$, $u$ are optical functions, which means that their level sets $\Cb_{\ub}$ and $C_u$ are incoming and outgoing null cones invariant under the $SO(3)$ action respectively. In the quotient spacetime, $\Cb_{\ub}$ and $C_u$ are then incoming and outgoing null lines respectively. We then denote
\begin{align*}
L=\frac{\partial}{\partial \ub},\ \Lb=\frac{\partial}{\partial u},
\end{align*}
and define the lapse function $\Omega$ by
\begin{align*}
-2\Omega^2=g(L,\Lb).
\end{align*}
Then the metric has the form
\begin{align*}
-2\Omega^2(\D\ub\otimes\D u+\D u\otimes\D\ub)+r^2\D\sigma_{\mathbb{S}^2}
\end{align*}
where the area radius function $r=r(\ub,u)$ is defined by
\begin{align*}
\text{Area}(S_{\ub,u})=4\pi r^2,
\end{align*}
and $\D\sigma_{\mathbb{S}^2}$ is the standard metric of the unit sphere.

\subsection{Equations} From the form of the metric, the unknowns of the Einstein-scalar field equations are $r$, $\Omega$ and the scalar field function $\phi$. What we really concern are their derivatives. We define the null expansions relative to the normalized pair of null vectors $\Omega^{-2}L$, $\Lb$ and the mass function $m$ by
\begin{align*}
h=\Omega^{-2}D r,\ \hb=\Db r,\ m=\frac{r}{2}(1+h\hb),
\end{align*}
where $D$ and $\Db$ are the restrictions on the orbit spheres of the Lie derivatives along $L$ and $\Lb$. When applying on functions, $D$ and $\Db$ are simply the ordinary derivatives. We then define the $D$ derivative of the lapse $\Omega$\begin{align*}
 \omega=D\log\Omega,
\end{align*}
while its $\Db$ derivative is not needed in this paper. Finally, we also need the derivatives of the scalar field function $\phi$:
\begin{align*}
L\phi=\frac{\partial}{\partial\ub}\phi,\ \Lb\phi=\frac{\partial}{\partial u}\phi.
\end{align*}

We then list below all the equations which are satisfied by the above quantities above and are needed in this paper. First of all, we have the following five null structure equations:\footnote{Readers may refer to \cite{Chr91} for the derivations of these equations, though the notations have slight differences. These equations can also be directly written down from the general null structure equations without symmetries, which can be found in the authors related paper \cite{Li-Liu} mentioned above. The derivations of these equations  in vacuum can be found in Christodoulou's work  \cite{Chr} on the formation of black holes.}
\begin{align}
\label{Dh}Dh=&-r\Omega^{-2}(L\phi)^2,\\
\label{Dbh}\Db(\Omega^2h)=&-\frac{\Omega^2(1+h\hb)}{r},\\
\label{Dhb}D\hb=&-\frac{\Omega^2(1+h\hb)}{r},\\
\label{Dbhb}\Db(\Omega^{-2}\hb)=&-r\Omega^{-2}(\Lb\phi)^2,\\
\label{Dbomega}\Db\omega=&\frac{\Omega^2(1+h\hb)}{r^2}-L\phi\Lb\phi.
\end{align}
The following two equations, which are equivalent, are the wave equation:
\begin{align}
\label{DbLphi}\Db(rL\phi)=&-\Omega^2h\Lb\phi,\\
\label{DLbphi}D(r\Lb\phi)=&-\hb L\phi.
\end{align}
Finally, we have the following equation about the mass function $m$:
\begin{align}
\label{Dm}Dm=&-\frac{1}{2}\hb\Omega^{-2}(rL\phi)^2,\\
\label{Dbm}\Db m=&-\frac{1}{2}h(r\Lb\phi)^2.
\end{align}

\section{A priori bounds for the solution}

We begin the proof of Theorem \ref{main}. Recall that we start from an arbitrary initial data $\alpha_0\in BV$ and the central line has a singular endpoint $e$, approaching which $\frac{2m}{r}\nrightarrow0$. The double null coordinate $(\ub,u)$ is chosen such that $\ub=0, u=-r$ on the boundary of the causal past of $e$, and  $u=u_0=-r_0$ and $\ub$ increases towards the future on the initial null cone $C_o$ where $r_0$ is the area radius of $\Cb_e\bigcap C_o$.

\subsection{Geometry on $\Cb_0$}

First of all we would like derive some identities on $\Cb_0$, the boundary of the causal past of $e$. We denote the restrictions on $\Cb_0$ of some geometric quantities, which are considered as functions of $u$:
$$\psi=\psi(u)=r\Lb\phi\Big|_{\Cb_0},\ \varphi=\varphi(u)=rL\phi\Big|_{\Cb_0},\ \Omega_0=\Omega_0(u)=\Omega\Big|_{\Cb_0},\ h_0=h_0(u)=h\big|_{\Cb_0}.$$
From $u=-r$ on $\Cb_0$, we must have $\hb\big|_{\Cb_0}\equiv-1$. Substituting this into \eqref{Dbhb}, we find
\begin{align}\label{Omega_0}
\frac{\partial}{\partial u}\log\Omega_0=-\frac{1}{2}\frac{\psi^2}{|u|},\ \text{and hence}\ -\log\frac{\Omega_0^2(u)}{\Omega_0^2(u_0)}=\int_{u_0}^u\frac{\psi^2(u')}{|u'|}\D u'.\end{align}
From \eqref{Dbh}, we have
\begin{align}\label{Omega_0^2h_0}
\frac{\partial}{\partial u}(\Omega_0^2h_0)=-\frac{\Omega_0^2(1-h_0)}{|u|},\ \text{and hence}\ -\log\frac{\Omega_0^2(u)h_0(u)}{\Omega_0^2(u_0)h_0(u_0)}=\int_{u_0}^u\frac{1}{|u'|}\left(\frac{1}{h_0(u')}-1\right)\D u'.
\end{align}
Because $m\big|_{\Cb_0}\ge0$ and the apparent horizon does not intersects $\Cb_0$, then $0<h_0\le 1$. Then we are going to prove an important lemma.
\begin{lemma}
Both $\Omega_0^2h_0$ and $\Omega_0$ are monotonically decreasing and converge to $0$ as $u\to0^-$. 
\end{lemma}
\begin{proof}
The monotonicity follows from the fact that the integrands in \eqref{Omega_0} and \eqref{Omega_0^2h_0} are positive. From Lemma 2 in \cite{Chr99} the integral in \eqref{Omega_0^2h_0} tends to infinity as $u\to0^-$. We rewrite this proof using the notations in this paper. Indeed, on $\Cb_0$, it holds $\frac{2m}{r}\big|_{\Cb_0}=1-h_0$, then using the fact that $m(0,u)$ is decreasing which follows from \eqref{Dbm}, we have
\begin{align*}
&\int_{3u}^u\frac{1}{|u'|}\left(\frac{1}{h_0(u')}-1\right)\D u'=\int_{3u}^u\frac{1}{|u'|}\frac{\frac{2m(0,u')}{|u'|}}{1-\frac{2m(0,u')}{|u'|}}\D u'\\
\ge&\int_{3u}^u\frac{1}{|u'|}\frac{\frac{|u|}{|u'|}\frac{2m(0,u)}{|u|}}{1-\frac{|u|}{|u'|}\frac{2m(0,u)}{|u|}}\D u'=\log\left(\frac{1-\frac{1}{3}\frac{2m(0,u)}{|u|}}{1-\frac{2m(0,u)}{|u|}}\right).
\end{align*}
If the integral in \eqref{Omega_0^2h_0} is bounded for all $u\in[u_0,0)$, then the first integral above should tend to zero when $u\to0^-$. However, this implies that $\frac{2m}{r}\to0$ from the above inequality, a contradiction. Therefore $\Omega_0^2h_0\to0$ as $u\to0^-$. If $\Omega_0\nrightarrow0$, then $\Omega_0$ has a positive lower bound because it is decreasing. Therefore $\Omega^2h_0\to0$ implies that $h_0\to0$. Substitute this to the equation in \eqref{Omega_0^2h_0}, we find as $u\to0^-$, $\Omega_0^2(u_0)h_0(u_0)-\Omega_0^2(u)h_0(u)=\int_{u_0}^u\frac{\Omega_0^2(1-h_0)}{|u'|}\D u'\to+\infty$ and hence $\Omega_0^2(u)h_0(u)\to-\infty$, a contradiction. We conclude that $\Omega_0\to0$ as $u\to0^-$ and the proof is completed.
\end{proof}
\begin{remark}
The proof in the following then depends only on the fact that $\Omega_0\to0$ monotonically. The infiniteness of the integral in \eqref{Omega_0^2h_0} depends strictly on the monotonicity of mass $m$ along $\Cb_0$. We do not expect a robust argument of this proof since the criteria $\frac{2m}{r}\nrightarrow0$ may not make sense beyond spherical symmetry. A robust version of this lemma, which is beyond reach right now, should include another suitable criteria in general case which is still an active area of research.
\end{remark}

\subsection{The a priori estimates} We are going to derive the a priori estimates for the geometric quantities. We fix a small constant $\delta>0$ and a constant $u_1\in(u_0,0)$ and denote
\begin{align*}
\mathscr{F}=\mathscr{F}(u_0,u_1)=\max\{1,\sup_{u_0\le u\le u_1}|\varphi(u)|\}.
\end{align*}
Without loss of generality, we also assume that $\Omega(u_0)\le1$. By the monotonicity of $\Omega_0$, we have $\Omega_0(u)\le1$ for all $u\in[u_0,0)$. Then we are going to prove
\begin{theorem}\label{estimate}
There exists a universal large constant $C_0\ge1$ such that the following statement is true. Suppose that 
\begin{align*}
\mathcal{A}=\mathcal{A}(\delta,u_0,u_1)=\max\{1,\sup_{0\le\ub\le\delta}\mathscr{F}^{-1}(|rL\phi(\ub,u_0)|+|u_0||\omega(\ub,u_0)|)\}<+\infty,
\end{align*}
and for some $C\ge C_0$ we have
\begin{align}\label{smallness}
C^2\delta|u_1|^{-1}\mathscr{F}\mathscr{W}^{\frac{1}{2}}\mathcal{A}\le1,\ \text{where}\ \mathscr{W}=\mathscr{W}(u_0,u_1)=\max\left\{1,\left|\log\frac{\Omega_0(u_1)}{\Omega_0(u_0)}\right|\right\}.
\end{align} Then we have the following estimates for $0\le\ub\le\delta$, $u_0\le u\le u_1$:\footnote{The notation $A\lesssim B$ means $A\le cB$ for some universal constant $c$.}
\begin{align}
\label{geometryestimate}\frac{1}{2}\Omega_0\le\Omega\le 2\Omega_0,&\ \frac{1}{2}|u|\le r\le 2|u|,\\
\label{estimate-Lphi}|rL\phi|\lesssim&\mathscr{F}\mathcal{A},\\
\label{estimate-Lbphi}|r\Lb\phi-\psi|\lesssim&\delta|u|^{-1}\mathscr{F}\mathcal{A},\\
\label{estimate-h}|h-h_0|\lesssim&\delta|u|^{-1}\Omega_0^{-2}\mathscr{F}^2\mathcal{A}^2,\\
\label{estimate-hb}|\hb+1|\lesssim&\delta|u|^{-1}\mathscr{F}\mathcal{A},\\
\label{estimate-omega}|u||\omega|\lesssim&\mathscr{F}\mathscr{W}^{\frac{1}{2}}\mathcal{A}.
\end{align}
Moreover, we have the improved estimate
\begin{align}\label{estimate-Lphiimp}
|rL\phi(\ub,u)-\varphi(u)|\lesssim|rL\phi(\ub,u_0)-\varphi(u_0)|+\delta|u|^{-1}\mathscr{F}^2\mathscr{W}^{\frac{1}{2}}\mathcal{A}^2.
\end{align}
\end{theorem}
\begin{proof}
We begin the proof by the following \emph{bootstrap assumptions}: 
\begin{align}
\label{bootstrapLphi}|rL\phi|\lesssim&C^{\frac{1}{4}}\mathscr{F}\mathcal{A},\\
\label{bootstrapLbphi}|r\Lb\phi-\psi|\lesssim&C^{\frac{1}{4}}\delta|u|^{-1}\mathscr{F}\mathcal{A},\\
\label{bootstraph}|h-h_0|\lesssim&C^{\frac{1}{2}}\delta|u|^{-1}\Omega_0^{-2}\mathscr{F}^2\mathcal{A}^2,\\
\label{bootstraphb}|\hb+1|\lesssim&C^{\frac{1}{4}}\delta|u|^{-1}\mathscr{F}\mathcal{A},\\
\label{bootstrapomega}|u||\omega|\lesssim&C^{\frac{1}{4}}\mathscr{F}\mathscr{W}^{\frac{1}{2}}\mathcal{A}.
\end{align}

From the equation \eqref{Dbomega} and \eqref{bootstrapomega}, we have
\begin{align*}
|\log\Omega-\log\Omega_0|\le\int_0^\delta|\omega|\D\ub\lesssim C^{\frac{1}{4}}\delta|u|^{-1}\mathscr{F}\mathscr{W}^{\frac{1}{2}}\mathcal{A}\lesssim C^{-1}.
\end{align*}
The last inequality is because of \eqref{smallness} and $|u|\ge|u_1|$\footnote{Because \eqref{smallness} is used very frequently in a similar manner, we will not point it out again when we use \eqref{smallness} in the rest of the paper. Because $\mathscr{W}\ge1$, \eqref{smallness} will also be used in the form $C^2\delta|u_1|^{-1}\mathscr{F}\mathcal{A}\le1$.}.  By choosing $C_0$ (and hence $C$) sufficiently large, we have
\begin{align*}
|\log\Omega-\log\Omega_0|\le \log 2
\end{align*}
and therefore \eqref{geometryestimate} holds for $\Omega$. Moreover, we have
\begin{align}\label{estimate-Omega-Omega0}
|\Omega-\Omega_0|\le\int_0^\delta\left|\Omega\omega\right|\D\ub\lesssim C^{\frac{1}{4}}\delta|u|^{-1}\Omega_0\mathscr{F}\mathscr{W}^{\frac{1}{2}}\mathcal{A}.
\end{align}

For $r$, we note that, from \eqref{bootstraph},
\begin{align}\label{estimate-h1}
|\Omega^2h|\lesssim\Omega_0^2\left(h_0+C^{\frac{1}{2}}\delta|u|^{-1}\Omega_0^{-2}\mathscr{F}^2\mathcal{A}^2\right)\lesssim\mathscr{F}\mathcal{A}.
\end{align}
The second inequality holds because $\Omega_0\le1\le \mathscr{F}\mathcal{A}$ by definition. We then use the equation $Dr=\Omega^2h$ to obtain\begin{align}\label{estimate-r-r0}
|r-|u||\le\int_0^\delta|\Omega^2h|\D\ub\lesssim\delta\mathscr{F}\mathcal{A}.
\end{align}
We then deduce that $|r-|u||\lesssim C^{-1}|u|$ and \eqref{geometryestimate} holds for $r$ if $C_0$ is sufficiently large\footnote{Similar to \eqref{smallness}, the estimates \eqref{geometryestimate} are used frequently and we will not point this out in the argument.}.

For $L\phi$, we consider the equation \eqref{DbLphi}. We write
\begin{align}\label{DbrLphi-varphi}
\frac{\partial}{\partial u}(rL\phi-\varphi)=-\left(\Omega^2hr^{-1}(r\Lb\phi)-\Omega_0^2h_0|u|^{-1}\psi\right).
\end{align}
Using \eqref{estimate-Lbphi}, \eqref{estimate-h}, \eqref{estimate-Omega-Omega0}, \eqref{estimate-h1} and \eqref{estimate-r-r0}, the right hand side can be estimated by
\begin{align*}
&|\Omega^2hr^{-1}(r\Lb\phi)-\Omega_0^2h_0|u|^{-1}\psi|\\
\lesssim&|\Omega^2-\Omega_0^2||h_0|u|^{-1}\psi|+|\Omega^2||h-h_0|||u|^{-1}\psi|+|\Omega^2h||r^{-1}-|u|^{-1}||\psi|+|\Omega^2hr^{-1}||r\Lb\phi-\psi|\\
\lesssim&C^{\frac{1}{4}}\delta|u|^{-1}\Omega_0^2\mathscr{F}\mathscr{W}^{\frac{1}{2}}\mathcal{A}\cdot|u|^{-1}|\psi|+C^{\frac{1}{2}}\delta|u|^{-1}\mathscr{F}^2\mathcal{A}^2\cdot|u|^{-1}|\psi|\\
&+\mathscr{F}\mathcal{A}\cdot\delta|u|^{-2}\mathscr{F}\mathcal{A}\cdot|\psi|+|u|^{-1}\mathscr{F}\mathcal{A}\cdot C^{\frac{1}{4}}\delta|u|^{-1}\mathscr{F}\mathcal{A}\\
\lesssim& C^{\frac{1}{2}}\delta|u|^{-1}\mathscr{F}^2\mathcal{A}^2\cdot|u|^{-1}\left(1+(1+\Omega_0^2\mathscr{W}^{\frac{1}{2}})|\psi|\right).
\end{align*}
Integrating the equation \eqref{DbrLphi-varphi}, we have
\begin{equation}\label{proof-estimate-Lphiimp}
\begin{split}
|rL\phi-\varphi|
\lesssim&|rL\phi-\varphi|\big|_{C_{u_0}}+C^{\frac{1}{2}}\delta|u|^{-1}\mathscr{F}^2\mathcal{A}^2\\
&+C^{\frac{1}{2}}\delta\mathscr{F}^2\mathcal{A}^2\left(\int_{u_0}^u\frac{(1+\Omega_0^2\mathscr{W}^{\frac{1}{2}})^2|\psi|^2}{|u'|}\D u'\right)^{\frac{1}{2}}\left(\int_{u_0}^u\frac{1}{|u'|^3}\D u'\right)^{\frac{1}{2}}\\
\lesssim&|rL\phi-\varphi|\big|_{C_{u_0}}+ C^{\frac{1}{2}}\delta|u|^{-1}\mathscr{F}^2\mathscr{W}^{\frac{1}{2}}\mathcal{A}^2
\end{split}
\end{equation}
where the second inequality is because of \eqref{Omega_0}. Then the estimate \eqref{estimate-Lphi} follows by \eqref{smallness}.

For $\Lb\phi$, we note that, from \eqref{bootstraphb}
\begin{align}\label{estimate-hb1}
|\hb|\lesssim1+C^{\frac{1}{4}}\delta|u|^{-1}\mathscr{F}\mathcal{A}\lesssim1+C^{-1}\lesssim1.
\end{align}
We then simply integrate the equation \eqref{DLbphi} and obtain, by \eqref{estimate-Lphi} we have proved above,
\begin{align*}
|r\Lb\phi-\psi|\lesssim\int_0^\delta|\hb L\phi|\D\ub\lesssim\delta|u|^{-1}\mathscr{F}\mathcal{A}.
\end{align*}

For $h$ and $\hb$, we use the equations \eqref{Dh} and \eqref{Dhb}. From \eqref{Dh}, using \eqref{estimate-Lphi}, we have
\begin{align*}
|h-h_0|\lesssim\int_0^\delta |r\Omega^{-2}(L\phi)^2|\D \ub\lesssim\delta\Omega_0^{-2}|u|^{-1}\mathscr{F}^2\mathcal{A}^2,
\end{align*}
which is the desired estimate \eqref{estimate-h}. Using \eqref{Dhb}, using \eqref{estimate-h1} and \eqref{estimate-hb1}, we have
\begin{align*}
|\hb+1|\lesssim\int_0^\delta\left|\frac{\Omega^2(1+h\hb)}{r}\right|\D\ub\lesssim\delta(1+\mathscr{F}\mathcal{A})\cdot|u|^{-1}\lesssim\delta|u|^{-1}\mathscr{F}\mathcal{A},
\end{align*}
which is the desired estimate \eqref{estimate-hb}.

For $\omega$, we use the equation \eqref{Dbomega}. The right hand side of \eqref{Dbomega} can be estimated by, using  \eqref{estimate-Lphi}, \eqref{estimate-Lbphi}, \eqref{estimate-h1} and \eqref{estimate-hb1},
\begin{align*}
\left|\frac{\Omega^2(1+h\hb)}{r^2}-L\phi\Lb\phi\right|\lesssim|u|^{-2}\mathscr{F}\mathcal{A}+|u|^{-2}\mathscr{F}\mathcal{A}(|\psi|+\delta|u|^{-1}\mathscr{F}\mathcal{A}).
\end{align*}
Integrating \eqref{Dbomega} and using \eqref{Omega_0}, we then have
\begin{align*}
|\omega|\lesssim|u|^{-1}\mathscr{F}\mathscr{W}^{\frac{1}{2}}\mathcal{A},
\end{align*}
which is the desired estimate \eqref{estimate-omega}.

Finally, we find estimates \eqref{estimate-Lphi}-\eqref{estimate-omega} we have proved improve the bootstrap assumptions \eqref{bootstrapLphi}-\eqref{bootstrapomega} if $C_0$ is sufficiently large. Therefore the estimates \eqref{estimate-Lphi}-\eqref{estimate-omega} hold without assuming \eqref{bootstrapLphi}-\eqref{bootstrapomega}. Using these estimates to repeat the derivation as in \eqref{proof-estimate-Lphiimp}, with the constant $C$ dropped, we know that \eqref{estimate-Lphiimp} also holds.

\end{proof}

\section{Instability theorems}

We then turn to the proof of the instability theorems. We divide the proof in two cases according to the behavior of $\varphi(u)$ as $u\to0^-$. The first case is the following.

\begin{theorem}\label{instability1}
If $\varphi(u)$ is unbounded as $u\to0^-$, then there exists two sequences $\delta_n\to0^+$ and $u_n\to0^-$ such that \eqref{maintheoremcondition} holds for $\ub=\delta_n,u=u_n$.
\end{theorem}
\begin{proof}
Because $\varphi(u)$ is unbounded, we can find a sequence $u_n\to0^-$ such that 
\begin{align*}
\varphi_n=\varphi(u_n)=\sup_{u_0\le u\le u_n}|\varphi(u)|\to\infty\ \text{as}\ n\to\infty.
\end{align*}
Define $\delta_n$ in terms of $u_n=-r_n$ by
\begin{align}\label{deltan}
\varphi_n^2=2^8c_1\Omega_n^4\log\frac{r_n}{4\Omega_n^2\delta_n}
\end{align}
where $\Omega_n=\Omega_0(u_n)$. It is obvious that $\delta_n\to0^+$ because $\Omega_n\to0$. We are going to prove such $\delta_n, u_n$ are two sequences we need.

We hope to apply Theorem \ref{estimate}, so we compute, for each $n$, 
\begin{align*}
C^2\delta_n|u_n|^{-1}|\varphi_n|\mathscr{W}_n^{\frac{1}{2}}=\frac{1}{4}C^2\Omega_n^{-2}\exp\left(-\frac{\varphi_n^2}{2^8c_1\Omega_n^4}\right)|\varphi_n|\mathscr{W}_n
\end{align*}
where $\mathscr{W}_n=\mathscr{W}(u_0,u_n)=\left\{1,\left|\log\frac{\Omega_n}{\Omega_0(u_0)}\right|\right\}$. We can see the right hand side tends to zero and therefore \eqref{smallness} holds for $\delta=\delta_n$, $u=u_n$ for sufficiently large $n$ depending on $C$ and the initial bound of $L\phi$ on $C_{u_0}$. As a consequence, we have the following estimates for a sufficiently large $C\ge C_0$ and $(\ub,u)\in[0,\delta_n]\times\{u_n\}$:
\begin{itemize}
\item $\displaystyle |\hb+1|\le C^{-1},\ \text{which implies}\ -\hb\ge\frac{1}{2}$.
\item $\displaystyle \Omega^{-2}\ge\frac{1}{4}\Omega_n^{-2}, 1\ge\frac{r}{2r_n}$.
\item $|rL\phi-\varphi_n|\le c|rL\phi(\ub,u_0)-\varphi(u_0)|+cC^{-1}|\varphi_n|$ for some $c$ depending on the initial bound of $L\phi$ on $C_{u_0}$ which follows from \eqref{estimate-Lphiimp} and implies that $\displaystyle |rL\phi|>\frac{1}{2}|\varphi_n|$ for $n$ sufficiently large.
\item $\displaystyle \Omega_n^2\delta_n\ge\Omega_n^2h_n\delta_n=\int_0^{\delta_n}\Omega_n^2h_n\D\ub\ge\frac{1}{4}\int_0^{\delta_n}\Omega^2h\D\ub=\frac{1}{4}(r-r_n)$, where we use $h_n=h(0,u_n)\ge h$ because of $Dh\le0$ from equation \eqref{Dh}.
\end{itemize}

From \eqref{Dm}, \eqref{deltan} and the above all estaimtes, we have, for $n$ sufficiently large,
\begin{align*}
m-m_n=&\frac{1}{2}\int_0^{\delta_n}(-\hb)\Omega^{-2}(rL\phi)^2\D\ub\\
>&\frac{1}{2^6}\delta_n\Omega_n^{-2}\varphi_n^2\\
=&\frac{1}{2^6}\delta_n\Omega_n^{-2}\cdot2^8c_1\Omega_n^4\log\frac{r_n}{4\Omega_n^2\delta_n}\\
\ge&\frac{c_1r}{2r_n}\cdot4\Omega_n^2\delta_n\log\frac{r_n}{4\Omega_n^2\delta_n}\\
\ge&\frac{c_1r}{2r_n}(r-r_n)\log\frac{r_n}{r-r_n}
\end{align*}
which is the inequality in \eqref{maintheoremcondition}. The last inequality above is because the function $x\log\frac{r_n}{x}$ is monotonically increasing for $x\in(0,4\Omega_n^2\delta_n]\subset(0,\frac{r_n}{\mathrm{e}}]$. Finally, $\frac{r-r_n}{r_n}\le c_0$ follows from
\begin{align*}
\frac{r-r_n}{r_n}\le\frac{4\Omega_n^2\delta_n}{|u_n|}=\exp\left(-\frac{\varphi_n^2}{2^8c_1\Omega_n^4}\right)
\end{align*}
since the right hand side tends to zero and hence not larger than $c_0$ if $n$ is sufficiently large. The proof is then completed.

\end{proof}

The second case is the following.
\begin{theorem}\label{instability2}
Suppose that $\varphi(u)$ is bounded by $\Phi\ge0$, and there exists some $\gamma\in(0,4)$ such that
\begin{align}\label{genericcondition}
\limsup_{u\to0^-}\Omega_0^{\gamma-4}(u)f(u;\gamma)>1
\end{align}
where the function $f$ is defined by
\begin{align*}
f(u;\gamma)=\frac{1}{\delta(u;\gamma)}\int_0^{\delta(u;\gamma)}|rL\phi(\ub,u_0)+(\varphi(u)-\varphi(u_0))|^2\D\ub
\end{align*}
and $\delta(u;\gamma)$ is defined in terms of $u$ by
\begin{align}\label{deltau}
\Omega^{4-\gamma}_0(u)=2^8c_1\Omega_0^4(u)\log\frac{|u|}{4\Omega_0^2(u)\delta(u;\gamma)}.
\end{align}
Then the conclusion of Theorem \ref{instability1} also holds.
\end{theorem}
\begin{proof}
From \eqref{genericcondition}, there exists a sequence $u_n\to0^-$ such that 
\begin{align}\label{genericcondition1}
f(u_n;\gamma)>\Omega^{4-\gamma}(u_n).
\end{align}
From \eqref{deltau}, we have $\delta_n=\delta(u_n;\gamma)\to0^+$ and
\begin{align*}
C^2\delta_n|u_n|^{-1}\Phi\mathscr{W}_n=C^2\frac{1}{4}\Omega_n^{-2}\exp\left(-\frac{1}{2^8c_1\Omega_n^{\gamma}}\right)\mathscr{W}_n.
\end{align*}
The right hand side tends to zero and therefore \eqref{smallness} holds for $\delta=\delta_n$, $u=u_n$ for $n$ sufficiently large. Then once we can prove that
\begin{align*}
\int_0^{\delta_n}|rL\phi(\ub,u_n)|^2\D\ub>\frac{1}{4}\delta_n\Omega_n^{4-\gamma}.
\end{align*}
the conclusion follows using the argument in the proof of Theorem \ref{instability1}.

To this end, we go back to equation \eqref{DbrLphi-varphi}. Integrating it on $C_{u_n}$ leads to
\begin{align*}
(rL\phi(\ub,u_n)-\varphi_n)-(rL\phi(\ub,u_0)-\varphi(u_0))=\int_{u_0}^u-\left(\Omega^2hr^{-1}(r\Lb\phi)-\Omega_0^2h_0|u'|^{-1}\psi\right)\D u'.
\end{align*}
The right hand side can be estimated similarly to the estimate of the error terms in \eqref{proof-estimate-Lphiimp}. Then we have, for $n$ sufficiently large,
\begin{align*}
|rL\phi(\ub,u_n)|\ge|rL\phi(\ub,u_0)+(\varphi_n-\varphi(u_0))|-c\delta|u_n|^{-1}\Phi^2\mathscr{W}_n^{\frac{1}{2}}
\end{align*}
for some constant $c$ depending on the initial bound of $L\phi$ on $C_{u_0}$. Now from \eqref{deltau} again,
\begin{align*}
c\delta|u_n|^{-1}\Phi^2\mathscr{W}_n^{\frac{1}{2}}=\frac{1}{4}\Omega_n^{-2}\exp\left(-\frac{1}{2^8c_1\Omega_n^{\gamma}}\right)\mathscr{W}_n^{\frac{1}{2}}\cdot c\Phi^2\le\sqrt{\frac{1}{4} \Omega^{4-\gamma}_n}
\end{align*}
if $n$ is sufficiently large, then from \eqref{genericcondition1},
\begin{align*}
\int_0^{\delta_n}|rL\phi(\ub,u_n)|^2\D\ub>\frac{1}{2}\delta_nf(u_n;\gamma)-\frac{1}{4}\delta_n\Omega_n^{4-\gamma}\ge\frac{1}{4}\delta_n\Omega_n^{4-\gamma}
\end{align*}
which is the desired inequality.
\end{proof}

\begin{remark}It is worth mentioning that in Christodoulou's original proof, when $\varphi(u)$ is bounded but not tends to zero, the conclusions of Theorem \ref{instability1} holds without any additional conditions like \eqref{genericcondition}. Indeed, the condition \eqref{genericcondition} is slightly different from that in Christodoulou's proof and we can see when $\varphi(u)$ is bounded but not tends to zero, \eqref{genericcondition} holds identically because $rL\phi(\ub,u_0)$ is of bounded variation and hence can be made right-continuous.
\end{remark}

The remaining part of the proof of Theorem \ref{main} is then similar to that in the last section in \cite{Chr99}. We still present the proof here for the sake of completeness.
\begin{proof}[Proof of Theorem \ref{main}] We fix the coordinate $\ub=r-r_0$ on $C_o=C_{u_0}$. Then
\begin{align*}
\alpha_0=\frac{\partial}{\partial r}(r\phi)=rL\phi\big|_{C_o}+\phi\big|_{C_o}.
\end{align*}
We denote $\theta_0=\theta_0(r)=rL\phi\big|_{\ub=r-r_0,u=u_0}$. As in  \cite{Chr99}, $\alpha_0$ being of bounded variation is equivalent to $\theta_0$ being bounded variation and $\frac{|\theta_0|}{r}\in L^1(0,+\infty)$. Therefore we consider instead $\theta_0$ in such a space. Suppose that $\theta_0\in\mathcal{E}$, then there exists a singular endpoint $e$ on $\Gamma$ and we have a double null coordinate $(\ub,u)$ relative to $e$ and in particular, $\Cb_e\bigcap C_o$ has area radius $r_0$. According to Theorem \ref{instability1} and \ref{instability2}, we have $\varphi(u)$ is bounded and 
\begin{align}\label{nogenericcondition}
\limsup_{u\to0^-}\Omega_0^{\gamma-4}(u)f(u;\gamma)\le1
\end{align}
for all $\gamma\in(0,4)$. We then define $f_1=f_1(r)$ such that it vanishes on $[0,r_0)$ and near infinity, and is absolutely continuous on $[r_0,+\infty)$ with $f_1(r_0)=1$. We also define $f_2=f_2(r)$ to be absolutely continuous on $[0,+\infty)$ such that it vanishes on $[0,r_0]$ and near infinity, and
\begin{align*}
f_2(r)=\sqrt{\frac{\D}{\D r}\left[(r-r_0)\Omega_0^2(u)\right]},\ r\in[r_0,r_0+1]
\end{align*}
where $u$ and $r$ are related through $r-r_0=\delta(u;\gamma=2)$ defined by \eqref{deltau}. Then for all $\gamma\in(0,4)$ and $\lambda_1\ne0$, we have
\begin{align}\label{limf1}
\lim_{u\to0^-}\frac{\Omega_0^{\gamma-4}(u)}{\delta(u;\gamma)}\int_0^{\delta(u;\gamma)}\lambda_1^2f_1^2(\ub+r_0)\D\ub=+\infty.
\end{align}
On the other hand, we define $u_{(\gamma)}$ through $\delta(u;\gamma)=\delta(u_{(\gamma)};2)$. From \eqref{deltau}, $\delta(u;\gamma)$ is increasing relative to $|u|$ and decreasing relative to $\gamma$. If $\gamma\in(0,2)$, we must have $|u|<|u_{(\gamma)}|$ and therefore $\Omega_0(u_{(\gamma)})>\Omega_0(u)$. We then have
\begin{equation}\label{lim1}
\begin{split}
&\lim_{u\to0^-}\frac{\Omega_0^{\gamma-4}(u)}{\delta(u;\gamma)}\int_0^{\delta(u;\gamma)}\lambda_2^2f_2^2(\ub+r_0)\D\ub\\
=&\lim_{u\to0^-}\lambda_2^2\Omega_0^{\gamma-4}(u)\Omega_0^2(u_{(\gamma)})\ge\lambda_2^2\lim_{u\to0^-}\Omega_0^{\gamma-2}(u)=+\infty.
\end{split}
\end{equation}
If $\gamma\in(2,4)$, we have $\Omega_0(u_{(\gamma)})<\Omega_0(u)$ and
\begin{equation}\label{lim2}
\begin{split}
&\lim_{u\to0^-}\frac{\Omega_0^{\gamma-4}(u)}{\delta(u;\gamma)}\int_0^{\delta(u;\gamma)}\lambda_2^2f_2^2(\ub+r_0)\D\ub\\
=&\lim_{u\to0^-}\lambda_2^2\Omega_0^{\gamma-4}(u)\Omega_0^2(u_{(\gamma)})\le\lambda_2^2\lim_{u\to0^-}\Omega_0^{\gamma-2}(u)=0.
\end{split}
\end{equation}
We then compute, when $\lambda_1\ne0$, for $\gamma\in(2,4)$, from \eqref{nogenericcondition}, \eqref{limf1}, \eqref{lim1} and \eqref{lim2},
\begin{align*}
&\limsup_{u\to0^-}\frac{\Omega_0^{\gamma-4}(u)}{\delta(u;\gamma)}\int_0^{\delta(u;\gamma)}|rL\phi(\ub,u_0)+\lambda_1f_1(\ub+r_0)+\lambda_2f_2(\ub+r_0)+(\varphi(u)-\varphi(u_0))|^2\D\ub\\
\ge&\liminf_{u\to0^-}\frac{\Omega_0^{\gamma-4}(u)}{2\delta(u;\gamma)}\int_0^{\delta(u;\gamma)}|\lambda_1f_1(\ub+r_0)|^2\D\ub\\
&-\limsup_{u\to0^-}\frac{\Omega_0^{\gamma-4}(u)}{\delta(u;\gamma)}\int_0^{\delta(u;\gamma)}|\lambda_2f_2(\ub+r_0)|^2\D\ub\\
&-\limsup_{u\to0^-}\frac{\Omega_0^{\gamma-4}(u)}{\delta(u;\gamma)}\int_0^{\delta(u;\gamma)}|rL\phi(\ub,u_0)+(\varphi(u)-\varphi(u_0))|^2\D\ub\\
=&+\infty.
\end{align*}
When $\lambda_1=0,\lambda_2\ne0$, we compute, for $\gamma\in(0,2)$, from \eqref{nogenericcondition} and \eqref{lim1},
\begin{align*}
&\limsup_{u\to0^-}\frac{\Omega_0^{\gamma-4}(u)}{\delta(u;\gamma)}\int_0^{\delta(u;\gamma)}|rL\phi(\ub,u_0)+\lambda_1f_1(\ub+r_0)+\lambda_2f_2(\ub+r_0)+(\varphi(u)-\varphi(u_0))|^2\D\ub\\
\ge&\liminf_{u\to0^-}\frac{\Omega_0^{\gamma-4}(u)}{2\delta(u;\gamma)}\int_0^{\delta(u;\gamma)}|\lambda_2f_2(\ub+r_0)|^2\D\ub\\
&-\limsup_{u\to0^-}\frac{\Omega_0^{\gamma-4}(u)}{\delta(u;\gamma)}\int_0^{\delta(u;\gamma)}|rL\phi(\ub,u_0)+(\varphi(u)-\varphi(u_0))|^2\D\ub\\
=&+\infty.
\end{align*}
This proves that $\theta_0+\lambda_1f_1+\lambda_2f_2\notin\mathcal{E}$ for all $\lambda_1,\lambda_2$ with $\lambda_1\ne0$ or $\lambda_2\ne0$. Now suppose that $\theta,\theta'\in\mathcal{E}$ and
$$\theta_{\lambda_1,\lambda_2}:=\theta_0+\lambda_1f_1+\lambda_2f_2\equiv\theta'_{\lambda'_1,\lambda'_2}:=\theta_0'+\lambda_1'f_1'+\lambda_2'f_2'.$$
Assume that $e'$ is the singular endpoint of $\Gamma$ in the maximal future development of $\theta_0'$ (and hence of $\theta'_{\lambda'_1,\lambda'_2}$) and $\Cb_{e'}\bigcap C_o$ has area radius $r_0'$. We then have $e=e'$ and $r_0=r_0'$. Because $f_i,f_i'$ vanish on $[0,r_0)$, we will have $\theta(r)\equiv\theta'(r)$ for $r\in[0,r_0)$ and hence the double coordinate $(\ub,u)$, the functions $\varphi(u)$ and $\Omega_0(u)$, are the same. Therefore $f_1\equiv f_1'$, $f_2\equiv f_2'$. We then write
\begin{align*}
\theta_{\lambda_1-\lambda_1',\lambda_2-\lambda_2'}\equiv\theta_0'.
\end{align*}
From the above argument, when $\lambda_1\ne\lambda_1'$ or $\lambda_2\ne\lambda_2'$, $\theta_{\lambda_1-\lambda_1',\lambda_2-\lambda_2'}\notin\mathcal{E}$ but $\theta_0'\in\mathcal{E}$. Therefore we must have $\lambda_1=\lambda_1'$ and $\lambda_2=\lambda_2'$. Finally, we conclude that $\theta\equiv\theta'$ and the proof is completed.
\end{proof}

\end{document}